\documentclass{article}
\usepackage{tikz}
\usepackage[margin=1in]{geometry}
\usepackage[mmddyy]{datetime}
\usepackage{amsfonts}
\usepackage{amsthm}
\usepackage{amsmath}
\usepackage{amssymb}
\usepackage{graphicx}
\usepackage[export]{adjustbox}
\usepackage{mathrsfs}
\usepackage{algorithm}
\usepackage[noend]{algpseudocode}
\usepackage{tabularx}
\usepackage{subfig}
\usepackage{hyperref}

\newcommand{\cover}{\mathrm{cover}}
\newcommand{\val}{\mathrm{value}}
\newcommand{\bh}{\mathbf{h}}
\newcommand{\inn}{\Gamma_{in}}
\newcommand{\outn}{\Gamma_{out}}
\newcommand{\inh}{\hat{\Gamma}_{in}}
\newcommand{\outh}{\hat{\Gamma}_{out}}
\newcommand{\cE}{\mathcal{E}}
\newcommand{\cG}{\mathcal{G}}
\newcommand{\chG}{\hat{\mathcal{G}}}
\newcommand{\chE}{\hat{\mathcal{E}}}
\newcommand{\chV}{\hat{\mathcal{V}}}

\newcommand{\chP}{\hat{\mathcal{P}}}

\tikzstyle{vertex}=[circle, draw, inner sep=0pt, minimum size=6ex, very thick]
\usetikzlibrary{graphs}

\newtheoremstyle{part}
    {3pt}
    {3pt}
    {}
    {\parindent}
    {\scshape}
    {:}
    {\newline} 
    {}

\newtheorem{thm}{Theorem}

\newtheorem{obs}[thm]{Observation}
\newtheorem{lemma}[thm]{Lemma}

\newtheorem{remark}[thm]{Remark}
\newtheorem{definition}[thm]{Definition}
\theoremstyle{part}

\theoremstyle{remark}

\newcommand{\cS}{\mathcal{S}}

\newcommand{\cA}{\mathcal{A}}
\newcommand{\cB}{\mathcal{B}}

\newcommand{\cV}{\mathcal{V}}

\newcommand{\vsk}{\vec{s}^*_k}
\newcommand{\vs}{\vec{s}}

\DeclareMathOperator*{\argmax}{arg\,max}

\newcommand{\greedyone}{\textsc{FullGreedy}}
\newcommand{\greedytwo}{\textsc{PartialGreedy}}
\newcommand{\hubalgo}{\textsc{HubHeuristic}}
\newcommand{\degalgo}{\textsc{DegreeHeuristic}}

\begin{document}

\title{On Greedy Approaches to Hierarchical Aggregation\thanks{This work is partially supported by NSF grant CCF-1657049 and NSF CAREER grant CCF-1844628. AP is partially supported by the National Science Foundation Graduate Research Fellowship under Grant No. DGE-1656518.}}
\author{
 Alexandra Porter\\
 \textit{Department of Computer Science} \\
\textit{Stanford University}\\
Stanford, CA \\
  \texttt{amporter@stanford.edu}
  \and
  Mary Wootters\\
  \textit{Departments of Computer Science} \\ \textit{and Electrical Engineering} \\
\textit{Stanford University}\\
Stanford, CA \\
  \texttt{marykw@stanford.edu}
}

\maketitle

\begin{abstract}
We analyze greedy algorithms for the  \em Hierarchical Aggregation  \em(HAG) problem, a strategy introduced in [Jia et al., KDD 2020] for speeding up learning on Graph Neural Networks (GNNs). The idea of HAG is to identify and remove redundancies in computations performed when training GNNs. The associated optimization problem is to identify and remove the \em most \em redundancies.
 
Previous work introduced a greedy approach for the HAG problem and claimed a 1-1/e approximation factor. We show by example that this is not correct, and one cannot hope for better than a 1/2 approximation factor. We prove that this greedy algorithm does satisfy some (weaker) approximation guarantee, by showing a new connection between the HAG problem and maximum matching problems in hypergraphs. We also introduce a second greedy algorithm which can out-perform the first one, and we show how to implement it efficiently in some parameter regimes. Finally, we introduce some greedy heuristics that are much faster than the above greedy algorithms, and we demonstrate that they perform well on real-world graphs. 
\end{abstract}

\section{Introduction}

In this work, we analyze an optimization problem that arises from \em Hierarchical Aggregration \em (HAG), a strategy that was recently introduced in \cite{HAG} for speeding up learning on \em Graph Neural Networks \em (GNNs).

At a high level, HAG identifies redundancies in the computations performed in training GNNs and elimates them.  This gives rise to an optimization problem, the \em HAG problem, \em which is to find and eliminate the most redundancies possible.  In this paper, we study greedy algorithms for this optimization problem.

Our contributions are as follows.
\begin{enumerate}
\item The work \cite{HAG} proposed a greedy algorithm, which we call \greedyone, for the HAG optimization problem, and claimed that it gives a $1 - 1/e$ approximation.  Unfortunately, this is not true, and we show by example that one cannot hope for a better than a $1/2$ approximation.  We prove a new approximation guarantee for \greedyone\  in Theorem~\ref{thm:mainapprox}.
In more detail, we are able to establish a $\frac{1}{d}(1 - 1/e)$ approximation ratio for a related objective function, where $d$ is a parameter of the problem ($d=2$ is a reasonable value). 

\item We propose a second greedy algorithm, \greedytwo, for the HAG optimization problem.  We show by example that this algorithm can obtain strictly better results than \greedyone\  mentioned above.  It is not obvious that \greedytwo\  is efficient, and in Theorem~\ref{thm:runningtime} we show that it can be implemented in polynomial time in certain parameter regimes.
\item While both of the greedy algorithms we study are ``efficient,'' in the sense that they are polynomial time, they can still be slow on massive graphs.  To that end, we introduce greedy heuristics and demonstrate that they perform well on real-world graphs. 
\end{enumerate}
Our approach is based on a new connection between the HAG problem and a problem related to maximum hypergraph matching.  We use this connection both in our approximation guarantees for \greedyone\ and in our efficient implementation of \greedytwo.

In Section~\ref{sec:prelim}, we define the HAG problem and set notation.
In Section~\ref{sec:algs}, we define algorithms \greedyone\ and \greedytwo.  
In Section~\ref{sec:efficiency}, we discuss the efficiency of these algorithms and show that both can be implemented in polynomial time in certain parameter regimes.
In Section~\ref{sec:approx}, we give a new approximation guarantee for \greedyone, and show by example that \greedytwo\ can do strictly better.
In Section~\ref{sec:experiments}, we compare \greedyone\ and \greedytwo\ in practice. We then discuss faster greedy heuristics and show empirically that they perform well.

\section{Preliminaries and Problem Definition}\label{sec:prelim}
\subsection{Abstraction of Graph Neural Networks}
Let $G = (V,E)$ be a directed\footnote{Throughout the paper we work with directed graphs, but if the underlying graph is undirected we may treat it as a directed graph by adding directed edges in both directions.} graph that represents some underlying data.  
For example, $G$ could arise from a social network, a graph of transactions, and so on.  The goal of a GNN defined on $G$ is to learn a \em representation \em $\mathbf{h}_v \in \mathbb{R}^s$ for each $v \in V$, with the goal of minimizing some loss function $\mathcal{L}(\{\mathbf{h}_v\,:\,v \in V\})$, which is typically designed so that the representations $\mathbf{h}_v$ can be used for prediction (for example, classifying unlabeled nodes).\footnote{A typical set-up for a GNN might be the following. The representations $\bh_v$ are some function $f$ of the features $\mathbf{x}_u$ and representations $\bh_u$ of the nodes $u$ in the neighborhood of $v$; a prediction $\mathbf{o}_v$ is a function $g$ of the $\mathbf{x}_v$ as well as of $\mathbf{h}_v$; and both $f$ and $g$ are fully connected feed-forward neural networks.  However, the details of GNNs will not actually matter for this work.} Graph neural networks were originally introduced by~\cite{scarselli2008graph}, and have numerous extensions and applications~\cite{hamilton2017inductive,kipf2016semi,velivckovic2017graph,ying2018hierarchical,cai2018comprehensive}.

Learning these representations $\mathbf{h}_v$ follows the abstract process depicted in Algorithm~\ref{alg:1}.  Each node $v$ calls a function \textsc{Aggregate} on the values $\mathbf{h}_w$ for $w \in \inn(v)$, resulting in an aggregated value $a_v$.  Here, $\inn(v)$ represents the set of nodes $w\in V$ so that $(w,v) \in E$.  Next, the node $v$ calls a function $\textsc{Update}$ on $a_v$ and the current value of $\bh_v$ to obtain an updated $\bh_v$.  Then this repeats.  Here, the function $\textsc{Aggregrate}$ can be as simple as a summation (e.g. in GCN~\cite{kipf2016semi}), or it can be more complicated (e.g. in GraphSAGE-P~\cite{hamilton2017inductive}).  In this work, we assume that \textsc{Aggregate} does not depend on the order of its inputs and can be applied hierarchically.  For example, we would have: 
\begin{align*}
&\Call{Aggregate}{\textsc{Aggregate}(x,y), \textsc{Aggregate}(z,w)} \\
&\qquad = \Call{Aggregate}{\textsc{Aggregate}(x,w), \textsc{Aggregate}(z,y)} \\
&\qquad = \Call{Aggregate}{x,y,z,w}. 
\end{align*}
This is often the case in GNNs (see \cite{HAG} for more details).

\begin{algorithm}
        \caption{Abstract GNN aggregation~\cite{HAG}}	\label{alg:1}
                \begin{algorithmic}[1]
		\Require Graph $G = (V,E)$, depth $K$
		\State Initialize $\bh_v^{(0)}$ appropriately.
		\Comment{Typically, set $\bh_v^{(0)}$ to the feature vector $\mathbf{x}_v$.}
		\For{$k =  1,...,K$}
		\For{ $v \in V$}
			\State $a_v^{(k)} \gets  \Call{Aggregate}{\{\bh_u^{(k-1)}|u \in \inn(v)\}}$
			\State $\bh_v^{(k)} \gets \Call{Update}{(a_v^{(k)},\bh_v^{(k-1)})}$
		\EndFor
		\EndFor
			\end{algorithmic}

	\end{algorithm}

\subsection{Hierarchical Aggregation}
The starting point for our work is the paper \cite{HAG}, which showed that there are significant improvements to be made (up to 2.8x, empirically), by cutting out redundant computations in Algorithm~\ref{alg:1}.  To see where redundant computations might arise, suppose that two nodes $u,v \in G$ have a large shared out-neighborhood $\outn(u) \cap \outn(v)$.  In Algorithm~\ref{alg:1}, we would call $\textsc{Aggregate}$ on the nodes $u$ and $v$ many times, once for each node in this shared out-neighborhood.
However, we can save computation by introducing an \em intermediate node \em $m$ so that $\inn(m) = \inn(u) \cap \inn(v)$ and $\outn(m) = \outn(u) \cap \outn(v)$, and then disconnecting $u$ and $v$ from its original shared out-neighborhood.  
Then, we only call \textsc{Aggregate} on $u$ and $v$ once, and we can use the stored computation many times.
This process is shown in Figure~\ref{fig:introex}.

 \begin{figure}
\centering
 \subfloat[\label{fig:introex1}]{\includegraphics[width=.2\textwidth]{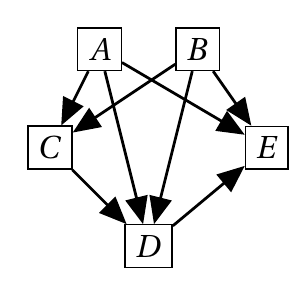}}
 \subfloat[\label{fig:introex2}]{\includegraphics[width=.3\textwidth]{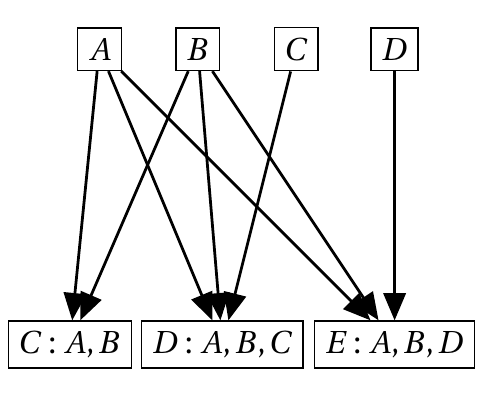}}
 \subfloat[\label{fig:introex3}]{\includegraphics[width=.3\textwidth]{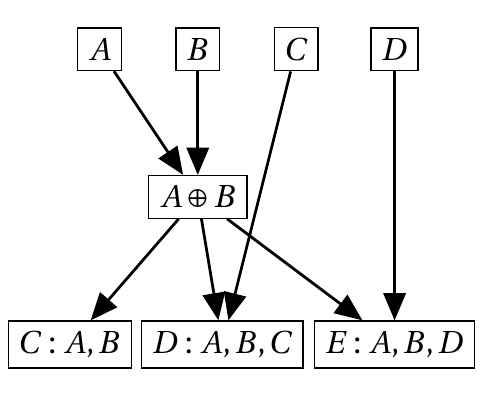}}

\caption{Example of hierarchical aggregation: (a) shows a directed graph $G = (V,E)$, (b) shows the GNN computation graph $\cG$, and (c) shows a possible HAG computation graph $\chG$, equivalent to $\cG$. 
The notation $[D:A,B,C]$ means that the node $D$ is requesting information from nodes $A,B,C$.  The notation $A \oplus B$ means that this intermediate node computes \textsc{Aggregate}$(A,B)$.
}\label{fig:introex}
\end{figure}

The Hierarchical Aggregation (HAG) problem is to find the best way to introduce such intermediate nodes.  We formally define the problem below.

\begin{definition}[GNN Computation Graph]
Given a directed graph $G = (V,E)$, the \textbf{GNN Computation Graph} $\cG$ for $G$ is a bipartite graph $(L,R,\cE)$, where $L$ and $R$ are copies of $V$, and, for $u \in L$ and $v \in R$, $(u_L,v_R) \in \cE$ if and only if $(u,v) \in E$.
We use $\inn(v)$ to denote the set of in-neighbors of a vertex $v$ in $\cG$, and we use $\outn(v)$ to denote the set of out-neighbors of $v$ in $\cG$.
\end{definition}

\begin{definition}[HAG Computation Graph]\label{def:hagcomputegraph}
Given a directed graph $G = (V,E)$, a \textbf{HAG Computation Graph} $\chG$ for $G$ is a graph $(\chV,\chE)$, where $\chV=L\cup M\cup R$ and $L$ and $R$ are copies of $V$.  $\chE$ contains directed edges from $L$ to $M$, from $M$ to $R$, and possibly within $M$, and the following property holds. For every directed edge $(u,v) \in E$, there is a unique directed path from $u_L$ to $v_R$ in $\chG$.
We use $\inh(v)$ to denote the set of in-neighbors of a vertex $v$ in $\chG$, and we use $\outh(v)$ to denote the set of out-neighbors edges of $v$ in $\chG$. 
When no edges in $\chE$ have both endpoints in $M$, so $\chG$ is tripartite, we call this a \emph{single-layer} HAG computation graph. When there exists integer $d$ such that for all $w \in M$, $|\inh(w)| = d$, then we call $\chG$ a \emph{d-HAG computation graph}.
\end{definition}
See Figure~\ref{fig:introex} for an example of a GNN computation graph and a HAG computation graph arising from a directed graph $G$.
We say that a GNN computation graph $\cG$ and a HAG computation graph $\chG$ are \textbf{equivalent} if they are both computation graphs for the same underlying graph $G$.  

We also use the $\cover$ function, as defined by~\cite{HAG}.
The cover of a vertex $v$ is just the set of all nodes in $L$ that eventually feed into it.
\begin{definition}
For a vertex $v$ in a HAG computation graph $\chG = (L \cup M \cup R, \chE)$, the \textbf{cover} of $v$ is defined as 
\[ \cover(v) = \{ w \in L \,:\, \text{there is a directed path from $w$ to $v$ in $\chG$}\}. \]
\end{definition}

We will subsequently  assume any HAG computation graph has the property that $\cover(m)$ is a distinct set for all distinct $m\in M$. 
This is without loss of generality, because if nodes $m_1$ and $m_2$ have the same cover, then they can perform the same function in the HAG graph and one of them can be removed.

Given a HAG computation graph, we can re-organize the computation in Algorithm~\ref{alg:1} in order to aggregate computations at the intermediate nodes in $M$.  This process is shown in Algorithm~\ref{alg:2}. Note that we need an ordering of $M$ such that for any $v \in M$, the vertices in $\inh(v)\cap M$ appear in the sequence before $v$. Since $\chG$ is a DAG, such a sequence can easily be constructed.
\begin{algorithm}
             \caption{Abstract GNN aggregation with added intermediate nodes~\cite{HAG}}
\label{alg:2}
                \begin{algorithmic}[1]
		\Require HAG Computation Graph $\chG = (\chV,\chE)$; depth $K$.
		\Require Sequence $\mathcal{M}=\{m_i\}_{i=1}^{|M|}$ for $M \subset \chV$ such that every $v\in M$ appears exactly once and after all its in-neighbors
		\State Initialize $\bh_v^{(0)}$ appropriately.
		\For{$k =  1,...,K$}
		\For{ $i=1,...,|M|$}
				\State $a_{m_i}^{(k)} \gets  \Call{Aggregate}{\{\bh_u^{(k-1)}|u \in \inh(m_i)\}}$
			\EndFor

		\For{ $v \in R$}
			\State $a_v^{(k)} \gets  \Call{Aggregate}{\{\bh_u^{(k-1)}|u \in \inh(v)\}}$
			\State $h_v^{(k)} \gets \Call{Update}{(a_v^{(k)},h_v^{(k-1)})}$
		\EndFor
		\EndFor
			\end{algorithmic}
	\end{algorithm}
	
In \cite{HAG}, the following cost function for a computation graph was considered.  We say that the cost of a computation graph $\cG$ with vertices $\mathcal{V}$ and right-hand side $R$ (either a HAG computation graph or a GNN computation graph) is 
\[ \mathrm{cost}(\cG) = c_{Agg}\sum_{w \in \mathcal{V}}(|\inn(w)|-1) + c_{Up}\cdot|R| \]
where $c_{Agg}$ and $c_{Up}$ are some constants representing the cost of an aggregation and an update respectively.
The reason for this cost function is that the cost to do an aggregation at a node $w \in V$ is proportional to the number of items in the aggregation, minus one. 
That is, one can ``aggregate'' a single item for free, and the cost grows linearly as we add more items.
The second term counts the cost of each update.  
We define to value of a HAG computation graph $\chG$ to be proportional to the amount of cost that it saves.
\begin{definition}
The \textbf{value} of a HAG Computation Graph $\chG = (\chV,\chE)$ is given by
\[ \val(\chG) = \frac{1}{c_{Agg}}\left(\mathrm{cost}(\cG) - \mathrm{cost}(\chG)\right) = \sum_{v \in M} \left[|\outh(v)|(|\cover(v)|-1)  -(|\inh(v)| - 1)\right]. \]
\end{definition}
To see that the two quantities are indeed equal, 
we may write
\begin{align*}
\frac{1}{c_{Agg}}\left(
\mathrm{cost}(\cG) - \mathrm{cost}(\chG)\right) &= \sum_{w \in R}(|\inn(w)| - 1) - \left( \sum_{v \in M} (|\inh(v)| - 1) + \sum_{w \in R}(|\inh(w)| - 1)\right)\\
&= \sum_{w \in R}\left[\left( \sum_{v \in \inh(w)}|\cover(v) |\right)- 1\right] -\left( \sum_{v \in M} (|\inh(v)| - 1) + \sum_{w \in R}(|\inh(w)| - 1)\right)\\
&= \sum_{w \in R}\left( \sum_{v \in \inh(w)\cap M}(|\cover(v) |-1)\right) -\left( \sum_{v \in M} (|\inh(v)| - 1) \right)\\
&= \sum_{v \in M} \left[|\outh(v)|(|\cover(v)|-1)  -(|\inh(v)| - 1)\right],
\end{align*}
where in the second line we have used the equivalence of $\cG$ and $\chG$ to say that $\inn(w)$ is equal to the disjoint union $\bigcup_{v \in \inh(w)} \cover(v)$, in the third we have combined summations over $w \in R$ and used the fact that $v \in \inh(w)\setminus M$ implies $v \in L$ and thus $|\cover(v)|=1$, and in the fourth we have switched the order of summations and used the fact that each $v$ in $\sum_{w \in R} \sum_{v \in \inh(w)}$ appears $|\outh(v)|$ times.
\subsection{The HAG Problem}
Given the above setup, we can formally define the HAG problem.  We additionally take two parameters $d$ and $k$.  The parameter $d$ is a bound on the left-degree of the aggregation nodes (for example, the work \cite{HAG} considered $d=2$ in their algorithm).  The parameter $k$ is a budget on the number of intermediate nodes allowed.

\begin{definition}[HAG problem]
Let $d$ be an integer.  The $d$-HAG problem is the following.  Given a graph $G$ and a node budget $k$, find a HAG computation graph $\chG = (\chV,\chE)$ for $G$ with the largest value, so that
$|M| \leq k$ and $|\inh(w)| = d$ for all $w \in M$.
\end{definition}

We also define a \em single-layer \em variation of the problem, which is to find the best way to add intermediate nodes in a way so that the resulting graph is tri-partite. The single layer variation is faster to compute and  we show empirically that single-layer solutions achieve almost as much value as general multi-layer solutions. 

\begin{definition}[single-layer HAG problem]
Let $d$ be an integer. The single-layer $d$-HAG problem is defined as the $d$-HAG problem with the additional constraint that $\chG$ be tripartite.
\end{definition}
We note that if $\chG$ is a single-layer $d$-HAG computation graph, value can be simplified: $value(\chG)=\sum_{v \in M} (|\outh(v)| - 1)(|\inh(v)| - 1)$.

\section{Greedy Algorithms}\label{sec:algs}
We study two natural greedy algorithms for the HAG problem.  We call these two algorithms \greedyone\ and \greedytwo.  Intuitively, \greedyone \ greedily choose an internal node, with all of its incoming and outgoing edges, and fixes it.  On the other hand, \greedytwo\ greedily chooses an internal node with all of its incoming edges, but re-optimizes the outgoing edges when each new internal node is added. That is, \greedyone\ is ``fully'' greedy in the sense that it makes a greedy choice for every edge, while \greedytwo\ is only ``partially'' greedy in the sense that it makes a greedy choice for the incoming edges, subject to fully optimizing over the outgoing edges.

To formally describe these algorithms, we define an additional function on HAG computation graphs.  
\begin{definition}
Given HAG computation graph $\chG=(\chV,\chE)$ with $\chV=  L\cup M \cup R$, for $X,Y \in \{L,M,R\}$ let $T_{\chG}(X,Y)$ denote the set edges in $\chG$ that either connect $X$ and $Y$ in $\chG$, or connect $Y$ to $Y$ in $\chG$:
\[
T_{\chG}(X,Y):=(\chE \cap (X\times Y)) \cup (\chE \cap (Y \times Y)).\]
\end{definition}

We begin with the algorithm \greedyone.  This algorithm was proposed by \cite{HAG}, and works as follows.  At each step, it chooses the internal node---complete with all ingoing and outgoing edges---that will increase $\val(\chG)$ by the most.  This is shown in Algorithm~\ref{alg:greedy1}.
\begin{algorithm}
                \caption{Greedy Algorithm \greedyone}\label{alg:greedy1}
                \begin{algorithmic}[1]
		\Require GNN Computation Graph $\cG = (L,R,\cE)$; aggregation node limit $k$, aggregation in-degree $d$.
		\State $M_0 \gets \emptyset$
		\State $\chE_0 \gets \cE$
		\State $\chG_i \gets (L \cup M_0\cup R,\chE_0)$
		\For{$i =  1,...,k$}
			\State $C \gets \argmax_{C\subseteq L\cup M_{i-1}\text{ s.t. }|C|=d} | \bigcap_{v\in C}\outh(v) \cap R|$

\Comment{Find the set $C$ of size $d$ to maximize the number of nodes in $R$ that request all of the nodes in $C$.}
			\State $R_C \gets \bigcap_{v\in C}\outh(v) \cap R$
			\State $M_i \gets M_{i-1} \cup \{v_i\}$ \Comment{add a new vertex $v_i$ to $M$}
			\State Construct the new edge set $\chE_i$:
		\begin{itemize}
			\item $\chE_i \gets \chE_{i-1}$
			\item Add edge  $(\ell, v_i)$ to $\chE_i$ for all $\ell \in C$.
			\item Add edge $(v_i, r)$ to $\chE_i$ for all $r \in R_C$.
			\item Remove any edges $(\ell, r)$ from $\chE_i$ with $\ell \in C$ and $r \in R_C$.
		\end{itemize}
			\State $\chG_i \gets (L\cup M_i\cup R,\chE_i)$
		\EndFor
			\end{algorithmic}
	\end{algorithm}

We next consider a greedy algorithm, \greedytwo, in which the edges between the intermediate nodes and receiving nodes are re-assigned at each iteration. In particular, at the $i^{th}$ step the edge set $T_{\chG_i}(M_i,R)$ is chosen to be optimal given $M_i$ and $T_{\chG_i}(L,R)$, rather than constructed by adding edges to the set $T_{\chG_{i-1}}(M_{i-1},R)$ from the previous step. Algorithm~\ref{alg:greedy2} describes this process. 
\begin{algorithm}
                \caption{Greedy Algorithm \greedytwo}\label{alg:greedy2}
                \begin{algorithmic}[1]
		\Require GNN Computation Graph $\cG = (L,R,\cE)$; aggregation node limit $k$, aggregation in-degree $d$.
		\State $M_0 \gets \emptyset$ 		\State $\chE_0 \gets \cE$
		\State $\chG_0 \gets (L \cup M_0\cup R,\chE_0)$
		\For{$i =  1,...,k$}
			\State Suppose that $\chG_{i-1}$ has vertices $\chV_{i-1} = L \cup M_{i-1} \cup R$.
			\For{$C \subseteq L\cup M_{i-1}$ s.t. $|C|=d$}\label{algline:inset}
				\State $M_i \gets M_{i-1} \cup \{v_i\}$
			\Comment{Add a new vertex $v_i$}
					\State 

$$ \mathcal{S}_C = \left\{\chG^{(C)}=(L \cup M_i \cup R,\chE^{(C)}): \begin{minipage}{7cm} \begin{center} $\chG^{(C)}$ is a d-HAG computation graph equivalent to $\cG$ and $T_{\chG^{(C)}}(L,M_i) = T_{\chG_{i-1}}(L,M_{i-1}) \cup (C \times \{v_i\})$ \end{center} \end{minipage} \right\}$$

\Comment{$\mathcal{S}_C$ is the set of all graphs $\chG^{(C)}$ that extend the left-hand side $T_{\chG_{i-1}}(L,M)$ of $\chG_{i-1}$ by adding an intermediate node $v$ with $\inn(v) = C$.}

					\State $\chG^{(C)}_{opt} \gets \argmax_{\chG^{(C)} \in \mathcal{S}_C} \val(\chG^{(C)})$  \label{algline:matchstep}
			\EndFor
			\State $\chG_i \gets \argmax_C \val(\chG^{(C)}_{opt})$

		\EndFor
			\end{algorithmic}
	\end{algorithm}
	\begin{remark}\label{rem:singlelayer}
	We note that both \greedyone\ and \greedytwo\ can be easily modified to find a single-layer solution.  In \greedyone\ (Algorithm~\ref{alg:greedy1}), we simply take the $\argmax$ over $L$ instead of $L \cup M_{i-1}$.  In \greedytwo\ (Algorithm~\ref{alg:greedy2}), 
we replace ``$C_j \subseteq L\cup M_{i-1}\text{ s.t. }|C_j| = d$" in Line~\ref{algline:inset} with ``$C_j \subseteq L\text{ s.t. }|C_j|=d$''. 
	\end{remark}

In the next two sections, we analyze the efficiency and approximation guarantees of both \greedyone\ and \greedytwo.

\section{Efficiency of \greedyone\ and \greedytwo}\label{sec:efficiency}
In this section, we discuss the efficiency of the two greedy algorithms presented above.
We note that \greedyone\ (Algorithm~\ref{alg:greedy1}) is clearly polynomial time if $d$ is constant.  In particular, the argmax can be naively implemented in time $O(n^d)$.
	
On the other hand,
it is not clear that \greedytwo\ (Algorithm~\ref{alg:greedy2}) is even polynomial time (in $n$), because it is not clear how to solve the optimization problem in line~\ref{algline:matchstep}.   However, we show that in fact this can be re-cast as a matching problem in hypergraphs, which is efficient in certain parameter regimes.
To do this, we need a few more definitions.

\begin{definition}\label{def:partialHAG}
Let $\chG = (L \cup M \cup R,\chE)$ be a HAG computation graph.  We define the \em partial HAG computation graph \em induced by $\chG$ to be $\chP = (L \cup M, T_{\chG}(L,M))$, the induced subgraph on the vertices $L \cup M$.

Given a partial HAG computation graph $\chP$, and a GNN computation graph $\cG = (L \cup R, \cE)$, let $\mathcal{S}(\chP,\cG)$ denote the set of HAG computation graphs $\chG$ on the vertices $L\cup M \cup R$, so that:
\begin{itemize}
\item[(a)] $\chG$ is equivalent to $\cG$, and
\item[(b)] $\chP$ is a partial computation graph induced by $\chG$.
\end{itemize}
\end{definition}

In this language, the $\argmax$ in Line \ref{algline:matchstep} of Algorithm~\ref{alg:greedy1} is maximizing over the set $\mathcal{S}(\chP^{(C)},\cG)$, where $\chP^{(C)}$ is the partial HAG computation graph induced by $\chG_{i-1}$ with an additional intermediate vertex $v$ with $\inh(v) = C$.

Below, we show that efficiently computing this $\argmax$ is equivalent to solving a hypergraph matching problem.

\begin{definition}\label{def:H}
Let $\cG = (L \cup R, \cE)$ be a GNN computation graph, 
and let $\chP$ be a partial HAG computation graph with vertices $L \cup M$.
Then for $r \in R$, define $H_r = H_r(\chP, \cG)$ to be the hypergraph with vertices $L$ and edges 
\[   \{ \cover(v)\, :\, v \in M \text{ and } \cover(v) \subseteq \inn(r) \}. \]
For an edge $e = \cover(v)$ of $H_r$, define the \em  weight \em of $e$ to be $|cover(v)|-1$.  

Let $H = H(\chP,\cG)$ be the disjoint union of the $H_r$, for $r \in R$.  (That is, the vertices of $H$ are $|R|$ disjoint copies of $L$, and the edges on the $r^{th}$ copy correspond to the edges in $H_r$.)
\end{definition}

\begin{lemma}\label{lem:reduction}\label{lem:bijection}
Let $\cG = (L \cup R, \cE)$ be a GNN computation graph, and 
let $\chP$ be a partial HAG computation graph.
Let $H = H(\chP,\cE)$ be as in Definition~\ref{def:H}.

Let $\mathcal{M}(H)$ denote the set of matchings in $H$.  Then there is a bijection 
\[ \varphi: \mathcal{M}(H) \to \mathcal{S}(\chP,\cG), \] 
so that for a matching $\mathcal{N} \in \mathcal{M}(H)$, 
\[ \val( \varphi(\mathcal{N}) ) = \val( \mathcal{N} ) - c(\chP), \]
where the value of a matching is defined as the sum of the weights of the edges in that matching, and where $c(\chP)$ is a constant that depends only on the partial HAG graph $\chP$.
When $\chP$ is a partial $d$-HAG graph with $k$ intermediate nodes, $c(\chP) = k(d-1)$.

In particular, if $\mathcal{N}$ is a maximum weighted hypergraph matching for $H$, then $\varphi(\mathcal{N})$ is a maximum value HAG computation graph in $\mathcal{S}(\chP,\cG)$.
\end{lemma}
\begin{proof}
We define the bijection $\varphi$ as follows.  Let $\mathcal{N}$ be a matching in $H$, and let $\mathcal{N}_r$ denote the restriction of $\mathcal{N}$ to $H_r$, recalling that $H$ is the disjoint union of $H_r$ for $r \in R$.
Suppose that the edges in $\mathcal{N}_r$ correspond to sets $\cover(v)$ for $v \in C_r$, for some set $C_r$.  (Notice that the edges in $\mathcal{N}_r$ will have this form by the definition of $H_r$.)  Then define $\varphi(\mathcal{N})$ to be the HAG computation graph $\chG$ so that the partial HAG computation graph induced by $\chG$ is $\chP$, and so that 
\begin{equation}\label{eq:added}
 \inh(r) = C_r \cup \left( \inn(r) \setminus \bigcup_{v \in C_r} \cover(v) \right) 
\end{equation}
for $r \in R$.  Notice that $\chP$ sets the edge structure between $L$ and $M$ and within $M$, so specifying $\inh(r)$ for each $r \in R$ completes the description of $\chG$. 

We now verify that $\chG = \varphi(\mathcal{N})$ is an element of $\mathcal{S}(\chP, \cG)$.  First, by construction it induces $\chP$ as a partial HAG graph.  
Second, $\chG = (L \cup M \cup R, \chE)$ is a HAG computation graph that is equivalent to $\cG = (L \cup R, \cE)$.  To see this, consider any edge $(\ell, r) \in \cE$.  We need to show that there is a unique path from $\ell$ to $r$ in $\chG$.  This is true because either $\ell$ is contained in exactly one set $\cover(v)$ for $v \in \inn(r)$, in which case the path is the one that goes through $v$; or $\ell$ is not in any sets $\cover(v)$, in which case the edge $(\ell, r)$ is added to $\chE$ by definition in \eqref{eq:added}.  It cannot be the case that $\ell$ is contained in $\cover(v)$ for multiple $v \in \inn(r)$, because $\mathcal{N}_r$ was a matching.

Next, we show that $\varphi$ is a bijection.  To see this, let $\chG \in \mathcal{S}(\chP, \cG)$.  Then observe that $\varphi^{-1}(\chG)$ is given by the matching $\mathcal{N}$ that is the disjoint union of matchings $\mathcal{N}_r$ for $r \in R$, so that $\mathcal{N}_r$ includes the edges $\cover(v)$ for $v \in \inh(r) \cap M$.

Finally, we establish the claim about the values of $\mathcal{N}$ and $\varphi(\mathcal{N})$.  
Let $\mathcal{N} = \varphi^{-1}(\chG)$ for some $\chG \in \mathcal{S}(\chP, \cG)$.
By the definition of the weights, and by the construction of $\mathcal{N}$, we have
\[ \val(\mathcal{N}) = \sum_r \sum_{v \in \inh(r) \cap M} ( |\cover(v)| - 1 ). \]
On the other hand, by the definition of the value of a HAG computation graph, we have
\begin{align*}
 \val(\chG) &= \sum_{v \in M} \left[|\outh(v)|(|\cover(v)|-1)  -(|\inh(v)| - 1)\right] \\
&= \sum_{v \in M}|\outh(v)|(|\cover(v)|-1) - \sum_{v \in M} (|\inh(v)|-1) \\
&= \sum_{r \in R} \sum_{v \in \inh(r) \cap M} (|\cover(v)|-1) - \sum_{v \in M}(|\inh(v)| - 1) \\
&= \val(\mathcal{N}) - c(\chP), 
\end{align*}
where we define $c(\chP) = \sum_{v \in M}(|\inh(v)| - 1)$, which we note depends only on the partial HAG graph $\chP$.
In particular, when $\chP$ is a partial $d$-HAG with $k$ intermediate nodes, $c(\chP) = k(d-1)$.
\end{proof}

As a corollary, when $d = O(1)$ is a constant, we see that \greedytwo\ (Algorithm~\ref{alg:greedy2}) can be implemented using a polynomial number of maximum weighted-hypergraph matching problems.  In particular, when $d=2$ or when $\mathrm{deg}(G)$ (the degree of the underlying graph) is constant, we can implement Algorithm~\ref{alg:greedy2} in polynomial time.

\begin{thm}\label{thm:runningtime}
Suppose that either:
\begin{itemize}
\item $d = 2$, and Algorithm~\ref{alg:greedy2} is restricted to a single layer (see Remark~\ref{rem:singlelayer}); or
\item $d = O(1)$ and $\deg(G) = O(1)$, where $\deg(G)$ is the maximum degree of the original graph $G$. 
\end{itemize}
Then Algorithm~\ref{alg:greedy2} can be implemented in polynomial time.
\end{thm}
\begin{proof}
When $d=2$ and \greedytwo\ is set to return a single layer graph, the associated hypergraph $H$ is just a graph with at most $n^2$ vertices and at most $kn$ edges; indeed, there are at most $n$ vertices and $k$ edges for each $H_r$, and $H$ is the disjoint union of the $H_r$ over at most $n$ vertices $r \in R$.
The problem of maximum weight matching in a graph can be solved using Edmond's algorithm in time $O(|V|^2 |E|)$ for a graph with $|V|$ vertices and $|E|$ edges.  Thus, by Lemma~\ref{lem:reduction}, the $\argmax$ in Line~\ref{algline:matchstep} of Algorithm~\ref{alg:greedy2} can done in time $O(n^3k)$.
Algorithm \ref{alg:greedy2} needs to call this algorithm $O(k \cdot (n+k)^2)$ times,  for each $i=1,\ldots,k$ and for each $C \subseteq L \cup M_{i-1}$ of size $d=2$.  Thus, the total running time is $O( k^2 (n + k)^2 n^3)$.

When $d > 2$ or \greedytwo\ is set to return a multi-layer graph, then the reduction from Lemma~\ref{lem:reduction} yields a weighted hypergraph maximum matching problem, which unfortunately is NP-hard.  However, when the degree $\mathrm{deg}(G)$ of the underlying graph (and hence of $\cG$) is a constant, then this decomposes into $n$ weighted hypergraph maximum matching problems, one for each $H_r$, and the number of vertices in $H_r$ is $|\inn(r)| \leq \mathrm{deg}(G) = O(1)$.   Therefore we can solve a maximum weighted hypergraph problem in $H_r$ by brute force in time $O(1)$. There are at most $n$ such problems, one for each $r$, and as above we solve each of them at most $k\cdot (n+k)^d$ times, yielding a running time of $O( k \cdot n \cdot (n + k)^d )$, where the $O(\cdot)$ notation is hiding dependence on $\deg(G)$.  
\end{proof}

\section{Approximation Guarantees for Single-Layer HAGs}\label{sec:approx}
In this section, we consider the approximation guarantees that can be obtained by \greedyone\ and \greedytwo.

\subsection{Approximation ratios for \greedyone}
We begin with \greedyone.
The work \cite{HAG} introduced \greedyone\ and claimed that it gives a $1 - 1/e$ approximation, in the sense that $\val(\chG_{greedy}) \geq \left(1 - \frac{1}{e}\right) \val(\chG_{opt})$, where $\chG_{opt}$ is the HAG computation graph of maximum value.  Unfortunately, as the example in Figure~\ref{fig:g1andsep} shows, this is not correct, and we cannot hope for better than a $1/2$ approximation.

\begin{figure}
\centering
 \subfloat[\label{fig:submodprob}]{\includegraphics[width=.3\textwidth]{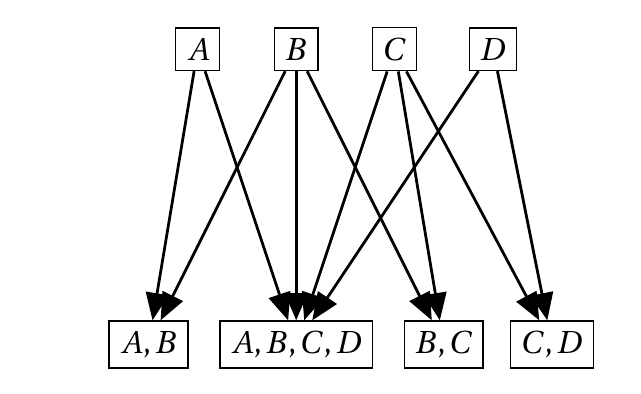}}
\ \ \ 
 \subfloat[\label{fig:submod1plus}]{\includegraphics[width=.3\textwidth]{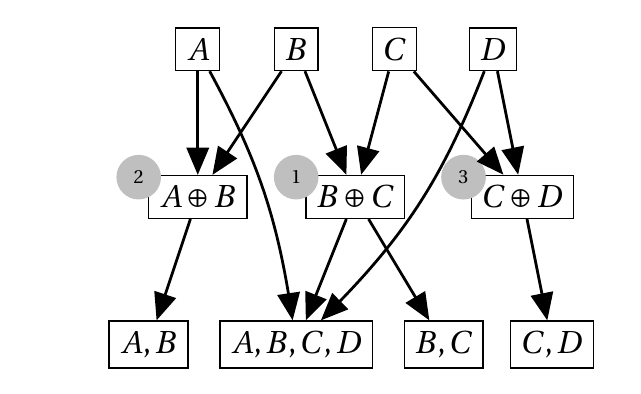}}
\ \ \ 
  \subfloat[\label{fig:submod2plus1}]{\includegraphics[width=.3\textwidth]{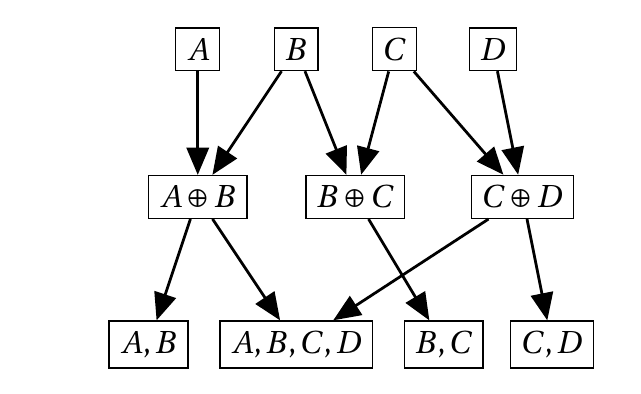}}

\caption{Example of a GNN computation graph (a) demonstrating that \greedyone\ cannot do better than a 1/2 approximation (in this example, for $k=3$) and that \greedytwo\ can strictly outperform \greedyone. The algorithm \greedyone\ will arrive at the solution shown in (b) by choosing the internal nodes in the order indicated: $B\oplus C$, $A\oplus B$, $C\oplus D$. The optimal solution for $k=3$ is shown in (c). The solution from \greedyone\ in (b) achieves a value of $1$ while the optimal solution has a value of $2$. There is a strict separation between \greedyone\ and \greedytwo\ because \greedytwo\ will reach the solution in (c) even if it chooses the internal nodes in the same order as \greedyone\ shown in (b).
}\label{fig:g1andsep}

\end{figure}

\begin{figure}
\centering
 \subfloat[\label{fig:submodprob}]{\includegraphics[width=.3\textwidth]{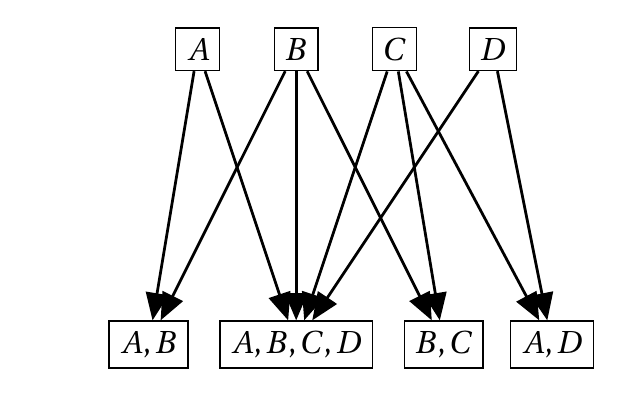}}
\ \ \ 
 \subfloat[\label{fig:submod1plus}]{\includegraphics[width=.3\textwidth]{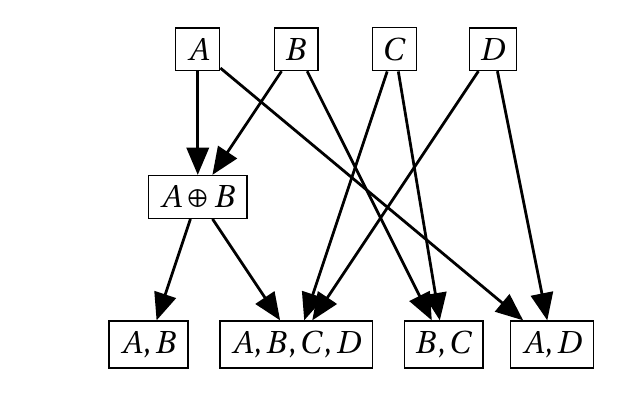}}
\ \ \ 
  \subfloat[\label{fig:submod2plus1}]{\includegraphics[width=.3\textwidth]{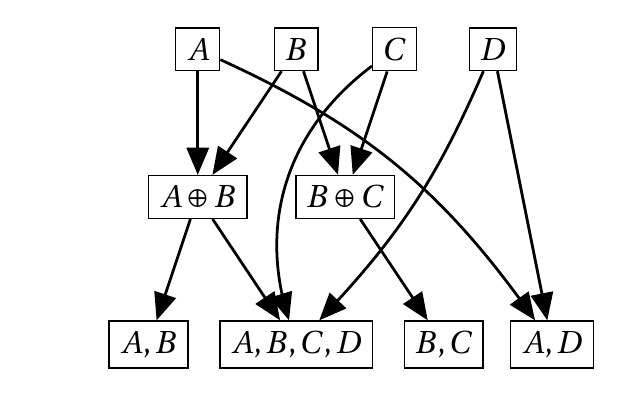}}\\
    \subfloat[\label{fig:submod2plus1}]{\includegraphics[width=.3\textwidth]{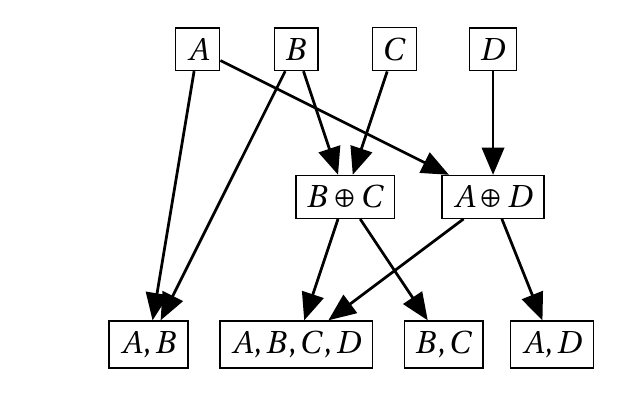}}
\ \ \
  \subfloat[\label{fig:submod2plus1}]{\includegraphics[width=.3\textwidth]{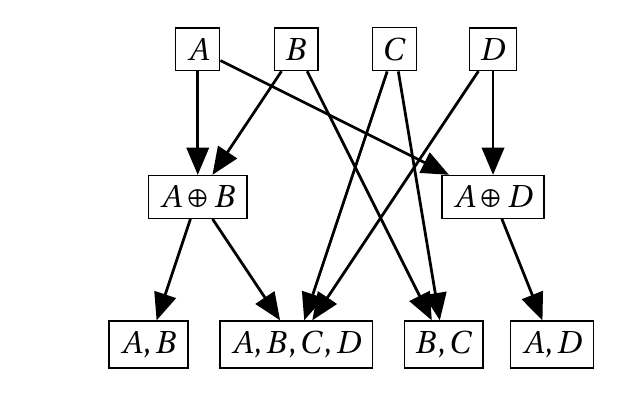}}
  \subfloat[\label{fig:submod2plus1}]{\includegraphics[width=.3\textwidth]{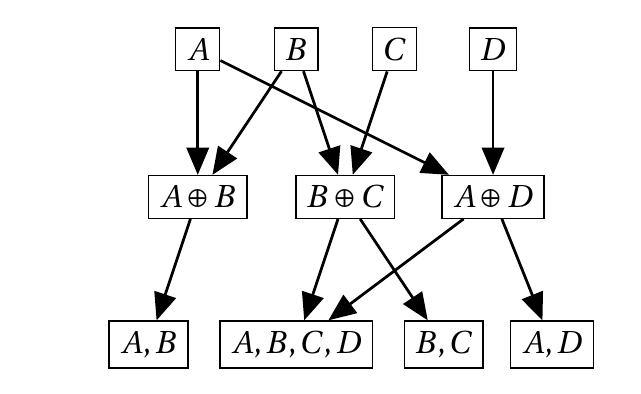}}

\caption{Example of a GNN computation graph (a) demonstrating that \greedytwo\ cannot do better than a 1/2 approximation (in this example, for $k=2$), and that the objective function for \greedytwo\ is not submodular. The algorithm \greedytwo\ will arrive at the solution in (c) by choosing $A\oplus B$ (as shown in (b)) and then $B \oplus C$, arriving at a value of $1$. The optimal solution for $k=2$ is shown in (d), and has a value of $2$. We see from (e) and (f) that the objective function for \greedytwo\ is not submodular, in the sense that adding an internal node $A\oplus D$ is more valuable after $A\oplus B$ and $B\oplus C$ have been added than when just $A \oplus B$ has been added. In (b) the value is $1$ and adding $A\oplus D$ to get (e) leaves the value at $1$. In (c) the value is $1$ and adding $A\oplus B$ to get (f) increases the value to $2$.
}\label{fig:notsubmod}
\end{figure}

In this section, we analyze \greedyone\ (Algorithm~\ref{alg:greedy1}), in the single-layer case.  Our main theorem is the following.
\begin{thm}\label{thm:mainapprox}
For a $d$-HAG computation graph $\chG$ with $k$ internal nodes, 
define 
\[ \widetilde{\val}(\chG) := \val(\chG) + k(d-1). \]
Then the single-layer $d$-HAG computation graph returned by \greedyone\ (Algorithm~\ref{alg:greedy1}, with a restriction to single-layer; see Remark~\ref{rem:singlelayer}) $\chG_{greedy}$ satisfies
\[ \widetilde{\val}(\chG_{greedy}) \geq \frac{1}{d}\left(1 - \frac{1}{e}\right) \widetilde{\val}(\chG^*), \]
where $\chG^*$ is the $d$-HAG computation graph with the largest value (and also the largest $\widetilde{\val}$).
\end{thm}
Unfortunately, we are not able to establish an approximation ratio for the function $\val(\cdot)$ itself, although we conjecture that a similar result holds.

The idea of the proof---which we give below in Section~\ref{sec:proof}---is as follows.
It is a standard result that greedy algorithms for submodular functions achieve a $1 - 1/e$ approximation ratio; this was the approach taken by \cite{HAG}.  Unfortunately, the \greedyone\ objective function is not technically submodular, since the order of the inputs matters, and this prevents the $1 - 1/e$ approximation result from being true.
However, we can use the connection to hypergraph matching developed in Lemma~\ref{lem:bijection} in order to translate the objective function of \greedyone \ to an objective function where the order does not matter, at the cost of a factor of $d$.  This results in a $\frac{1}{d}(1 - 1/e)$ approximation ratio for $\widetilde{\val}$.

\subsection{\greedytwo\ can strictly outperform \greedyone}
We first observe by example that \greedytwo\ also cannot achieve an approximation ratio better than $1/2$: the example is given in Figure~\ref{fig:notsubmod}.  Notice that this example also shows that the objective function that \greedytwo \ is greedily optimizing is not submodular.

However, we also show by example that there are graphs for which \greedytwo\ is strictly better than \greedyone.  Indeed, an example is shown in Figure~\ref{fig:g1andsep}.

Thus, the algorithm that runs both \greedyone\ and \greedytwo\ and takes the better of the two achieves at least the approximation guarantee of Theorem~\ref{thm:mainapprox}, and can sometimes do strictly better than \greedyone.


\subsection{Proof of Theorem~\ref{thm:mainapprox}}\label{sec:proof}
In this section we prove Theorem~\ref{thm:mainapprox}.
Since we are consider single-layer graphs, we can simplify the notation somewhat.
Let 
\[ \cA_d = \{ s \subset L : |s| = d \} \]
be the set of subsets of size $d$; we will associate each such subset with a possible intermediate node $m \in M$, so that $\inh(m) = s$.
Let 
\[ \cB_{d,k} = \{ \{s_1, s_2, \ldots, s_k\} : s_i \in \cA_d \ \forall i \} \]
be the collection of all ways to choose $k$ sets $s \in \cA_d$.  Thus, an element $S = \{s_1, s_2, \ldots, s_k\} \in \cB_{d,k}$ represents a set of possible solutions to the single-layer $d$-HAG problem, where the intermediate nodes are $m_1, \ldots, m_k$ so that $\inh(m_i) = s_i$.

\begin{remark}\label{rem:abuse}
With the above connection in mind, we will abuse notation and say that ``$\chP$ is the partial HAG computation graph induced by $S = \{s_1, \ldots, s_k\}$ and $\cG$,'' when we mean that $\chP$ is induced by a HAG computation graph $\chG$ that is equivalent to $\cG$ and whose intermediate nodes $m_1, \ldots, m_k$ have $\inh(m_i) = s_i$.  
\end{remark}

We first define the sequence of HAG graphs chosen by this algorithm.

\begin{definition}\label{def:addgreedhagseq}
Let $\cG$ be a GNN computation graph.
Let $s_1, s_2, \ldots, s_k \in \cB_{d,k}$.
Define the \emph{greedy d-HAG sequence} of HAG computation graphs $\chG_1, \ldots, \chG_k$ to be the sequence of graphs that arise when we greedily assign edges between $M$ and $R$ while inserting the internal nodes corresponding to $s_1, \ldots, s_k$ in that order.  That is, we define
$\chG_0 = \cG$, and given $\chG_{i-1} = (L \cup M_{i-1} \cup R, \chE_{i-1})$, we recursively 
define $\chG_i$ as follows.  

Let $\chG'_i = (L \cup M_i \cup R, \chE_i')$, where $M_i = M_{i-1} \cup \{v_i\}$, and $\chG_i' = \chE_{i-1} \cup \{(u,v_i) \,|\, u \in s_i \}$.
Now let $\chP_i$ denote the partial HAG computation graph induced by $\chG_i'$ and $\cG$
(as per Definition~\ref{def:partialHAG}), and define
\[ \chG_i = \argmax_{ \substack{\chG \in \cS(\chP,\cG) \\ T_{\chG_{i-1}}(M_{i-1},R) \subseteq T_{\chG}(M_i,R)} } \val(\chG). \]
\end{definition}

\begin{remark}\label{rem:greedyseq}
Let $\chG_i = (L\cup M_i \cup R, \chE_i)$ be the $i^{th}$ graph in the greedy d-HAG sequence.
Then we obtain $\chG_i$ from $\chG_{i-1}$ by (a) 
adding an internal vertex $v_i$ with $\inh(v_i) = s_i$, and (b) for each $r \in R$, 
greedily adding the edge $(v_i, r)$ if we can; that is, if $\cover(v_i) \subseteq (\inh(r) \cap L)$.  (And if we do that, we remove any edges between $\cover(v_i)$ and $r$).
\end{remark}


Before we proceed, we set some notation that will be helpful for the rest of the proof.
\begin{definition}\label{notation}
We will denote a length-$i$ ordered sequence $(s_1, \ldots, s_i) \in \cA_d^i$ by $\vs_i$.
Throughout, $S^* \in \cB_{d,k}$ will denote an element corresponding to an optimal solution $\chG^*$ to the single-layer $d$-HAG problem; that is, the intermediate nodes $M$ of an optimal solution $\chG^*$ define $S^*$ by
$S^* = \{ \inh(m) : m \in M \}$.  We will order the elements of $S^*$ arbitrarily as $(s_1^*, \ldots, s_k^*)$, and denote a prefix $(s_1^*, \ldots, s_i^*)$ by $\vs^*$.
We will use $(\vs, \vec{s'})$ to denote concatenation
e.g. $(\vs_i,\vec{s^*_i}) = (s_1,...,s_i,s^*_1,...,s_i^*)$. 
\end{definition}

With this notation, we have the following definition.
\begin{definition}\label{def:ocost}
For some GNN computation graph $\cG=(\cV,\cE)$, 
we define the functions $h: \cA_d^k \to \mathbb{Z}^+$ and $f:\cB_{k,d} \to \mathbb{Z}^+$ as follows.
The \emph{ordered matching value} function $h$ is defined as
\[ h(\{s_i\}_{i=1}^j) =(d-1) \sum_{i=1}^j  |\outh^{(j)}(m_i)|,\] 
where $\chG_1,...,\chG_k$ is the additive greedy d-HAG sequence, $\outh^{(j)}(m_i)$ is the out-neighborhood of $m_i$ in $\chG_j$, and $m_i$ is the vertex in $M$ in $\chG_j$ with $\inn(m_i) = s_i$. Now let $\outh^{(j)}$ be defined with respect to the graph 
$\chG = (L,M,R, \chE)$ that is the maximum value HAG computation graph in  $\cS(\chP, \cG)$, so that $\chP$ is the partial HAG computation graph induced by $S_j$ and $\chG$ (c.f. Remark~\ref{rem:abuse}), and let $M = \{m_1, \ldots, m_j\}$.  Then
the \emph{ maximum matching value} function $f$ is defined as \[ f(S_j) = (d-1)\sum_{i=1}^j |\outh^{(j)}(m_i)|.\]

%
\end{definition}

The functions $h$ and $f$ are related by an additive term of $(d-1)k$ to the values of various graphs, as shown below in Lemma~\ref{lem:hf}.
We use them instead of these values, because as per Lemma~\ref{lem:bijection}, we will see that they correspond directly to the size of the matchings in a hypergraph.  

\begin{lemma}\label{lem:hf}
Let $\cG$ be a GNN computation graph.
For any $\vs_j \in \mathcal{A}_d^j$, let $\chG_j$ be the $j^{th}$ graph in the greedy d-HAG sequence defined by $\vs_j$ and $\cG$.  Let $\chG^*$ be the maximum-value element of $\mathcal{S}(\chP, \cG)$, 
where $\chP$ is the partial d-HAG graph induced by $S_j = \{s_1, \ldots, s_j\} \in \mathcal{B}_{d,j}$ (c.f. Remark~\ref{rem:abuse}).
Then
\[ \val(\chG_j) = h(\vs_j) - (d-1)j \]
and
\[ \val( \chG^* ) = f(S_j) - (d-1)j. \]
In particular, $h(\vs_j) = \widetilde{\val}(\chG_j)$ and $f(S_j) = \widetilde{\val}(\chG^*)$.
\end{lemma}
\begin{proof}
 For the first expression, let $\outh(m_i)$ denote the out-neighborhood of $m_i$ in $\chG_j$, where $m_i$ is the vertex in $M$ in $\chG_j$ with $\inn(m_i) = s_i$.
Then using the fact that $|\inh(m_i)| = d$ and $\cover(m_i)=\inh(m_i)$ for all $i$,
(recall that we are working in a single-layer $d$-HAG) we have
\begin{align*}
\val(\chG_j) &= 
\sum_{i=1}^j (|\inh(m_j)| - 1)(|\outh(m_j)| - 1) \\
&=
\sum_{i=1}^j\left[ |\outh(m_i)|\cdot d - \left(|\outh(m_i)|+d-1\right)\right] \\
&= \sum_{i=1}^j\left[ |\outh(m_i)|\cdot (d-1) -(d-1)\right] \\
&= (d-1)\sum_{i=1}^j\left[ |\outh(m_i)|-1\right] =(d-1)\sum_{i=1}^j\left[ |\outh(m_i)|\right]-(d-1)j\\
&= h(\vs_j) - (d-1)j.
\end{align*}
Similarly,
let $\outh^{*}$ be with respect to the graph $\chG^*$.
Then again using that $|\inh^{*}(m_i)| = d$ for all $i$, we have
\begin{align*} 
\val(\chG) &=  \sum_{i=1}^j\left[ |\outh^{*}(m_i)|\cdot d - \left(|\outh^{*}(m_i)|+d-1\right)\right]\\
&=(d-1)\sum_{i=1}^j\left[ |\outh^{*}(m_i)|\right]-(d-1)j  \\
&= f(\{s_1, \ldots, s_j\}) - (d-1)j. 
\end{align*}

\end{proof}

\begin{obs}\label{obs:monotone}
 The function $f$ is monotone.
\end{obs}
\begin{proof}
By Lemma~\ref{lem:hf} it suffices to show that 
\[ \max_{\chG \in \mathcal{S}(\chP, \cG) } \val(\chG) \]
does not decrease when $\chP$ goes from being the partial d-HAG graph induced by $\{s_1, \ldots, s_j\}$ to the partial d-HAG graph induced by $\{s_1, \ldots, s_j, s_{j+1}\}$.  This is true because the set $\mathcal{S}(\chP, \cG)$ only grows larger with this change, and so the maximum is being taken over a larger set.
\end{proof}

\begin{lemma}\label{lem:orderedmatch}
Let $\cG$ be a GNN graph.
Let $S_t \in \cB_{d,k}$.
Then for any ordering $\vs_t= s_1,...,s_t$ of $S_t$:
 \[\frac{1}{d}\cdot f(S_t) \leq h(\vec{s_t}) \leq f(S_t)\] 
  \end{lemma}
  \begin{proof}
Let $\chP$ be the partial d-HAG graph induced by $S_t$ and $\cG$ (c.f. Remark~\ref{rem:abuse}), and let $\mathcal{S} = \mathcal{S}(\chP, \cG)$.  Let $\chG_1, \ldots, \chG_t$ be the greedy d-HAG sequence defined by $\vs_t$ and $\cG$.  
Let $H^{(i)}$ be the hypergraph associated with $\chG_i = (L \cup M_i \cup R, \chE_i)$ as in Definition~\ref{def:H}.
Consider the bijection $\varphi$ from (the proof of) Lemma~\ref{lem:bijection}, and let $\mathcal{N}^{(i)}$ be a matching in  $H^{(i)}$, so that
 $\varphi(\mathcal{N}^{(i)}) = \chG_i.$
Recall that the matching $\mathcal{N}^{(i)}$ can be decomposed into matchings $\mathcal{N}^{(i)}_r$, each on the graph $H_r$ from Definition~\ref{def:H}.
In more detail, the proof of Lemma~\ref{lem:bijection} shows that the hyperedge $(s_j \cap \inn(r))$ is in $\mathcal{N}_r^{(i)}$ if and only if the edge $(m_j, r)$ is in $\chG_i$.

First, we observe by Lemma~\ref{lem:bijection} and Lemma~\ref{lem:hf} that for any $i \leq t$ and for any $\vs_i \in \mathcal{A}_d^i$,
\begin{equation}\label{eq:Nsize}
 h(\vs_i) = (d-1) \cdot \sum_r |\mathcal{N}_r^{(i)}| = \val(\mathcal{N}^{(i)}),
\end{equation}
where the value on the right hand side represents the (weighted) value of the matching.  (Notice that since we are looking at the single-layer d-HAG problem, all weights are equal to $d-1$).

Similarly, let $\mathcal{N}^*$ be such that $\varphi(\mathcal{N}^*) = \chG^*$, where $\chG^*$ is the maximum-value element of $\mathcal{S}(\chP, \cG)$ where $\chP$ is induced by $S_t$.
Lemma~\ref{lem:bijection} implies that $\mathcal{N}^*$ is a maximum hypergraph matching for $H^{(t)}$.
As above, by the definition of $H$, $\mathcal{N}^*$ decomposes into matchings $\mathcal{N}_r^*$ of $H^{(t)}_r$ for each $r \in R$.
Then for $S_t \in \mathcal{B}_{d,t}$, Lemma~\ref{lem:bijection} and Lemma~\ref{lem:hf} imply that
\begin{equation}\label{eq:Noptsize}
f(S_t) = (d-1) \cdot \sum_r |\mathcal{N}^*_r| = \val(\mathcal{N}^*).
\end{equation}

Now consider the change from $\mathcal{N}_r^{(i)}$ to $\mathcal{N}_r^{(i+1)}$.  When we pass from $H^{(i)}$ to $H^{(i+1)}$, we add a hyperedge $e_r := s_i \cap \inn(r)$ to each graph $H_r^{(i)}$. The hyperedge $e_r$ is added to the matching $\mathcal{N}_r^{(i+1)}$ if and only if it can be: that is, if and only if it does not intersect $s_j \cap \inn(r)$ for some $j < i$.  This is because of the definition of the correspondence $\varphi$, and also the observation in Remark~\ref{rem:greedyseq} about how $\chG_{i+1}$ is created from $\chG_i$.  

Therefore, for any $r \in R$, the matching $\mathcal{N}_r^{(t)}$ can be found by the following algorithm:
\begin{itemize}
\item Let $H_r^{(t)}$ be as above.
\item $\mathcal{N}^{(0)}_r = \emptyset$
\item For $i=1,\ldots,t$:
\begin{itemize}
	\item If the hyperedge $s_i \cap \inn(r)$ can be added to $\mathcal{N}^{(0)}_r$ and still form a hypergraph matching of $H_r^{(t)}$, then let $\mathcal{N}^{(i)}_r = \mathcal{N}^{(i)}_r \cup \{ s_i \cap \inn(r) \}$.
\end{itemize}
\end{itemize}
We observe that this is the classical greedy algorithm for maximum hypergraph matching.  This algorithm is well-known to achieve an approximation ratio of $1/d$~\cite{chandra2001greedy}.  That is,
\[ \frac{1}{d} \val(\mathcal{N}^*) \leq \val(\mathcal{N}^{(t)}) \leq \val(\mathcal{N}^*). \]
By \eqref{eq:Nsize} and \eqref{eq:Noptsize}, this implies that
\[ \frac{1}{d} f(S_t) \leq h(\vs_t) \leq f(S_t), \]
as desired.

 \end{proof}

\begin{lemma}\label{lem:matchingdiff}
Let $S^*$ be as in Definition~\ref{notation}.
Let $\vsk = (s^*_1,...,s^*_k)$ be any order of elements of $S^*$. 
Let $\vs_i = (s_1, \ldots, s_i)$ be the nodes added after $i$ steps of \greedyone.  
Then
$$h((\vs_i,\vsk)) -h(\vsk) \geq - \frac{d-1}{d}h(\vsk).$$
\end{lemma}
\begin{proof}
Letting $S_i$ denote the set of elements of $\vs_i$, we have
\[h((\vs_i,\vsk)) -h(\vsk) \geq \frac{1}{d}f(S_i\cup S^*) - h(\vsk) \geq \frac{1}{d}f(S^*)  -h(\vsk) \geq \frac{1}{d}h(\vsk)-h(\vsk) = -\frac{d-1}{d}h(\vsk)\]
The first inequality is  an application of Lemma~\ref{lem:orderedmatch}. The second inequality follows from $f$ being monotone (Observation~\ref{obs:monotone}). The third inequality is because $f(S^*)$ gives the optimal graph choice given $S^*$, while $h(\vsk)$ gives one option of graph choice given $S^*$.
\end{proof}

\begin{lemma}\label{lem:submod2}
Let $S^*$ be as in Definition~\ref{notation}.
Let $\vsk = (s^*_1,...,s^*_k)$ be any order of elements of $S^*$. 
Let $\vs_i = (s_1, \ldots, s_i)$ be the nodes added after $i$ steps of \greedyone.  
Then, we have
\[ h((\vs_i,\vsk)) - h(\vs_i) \leq \left(1-\frac{1}{k+1}\right) (h((\vs_i,\vsk))-h(\vs_{i-1})). \]
\end{lemma}
\begin{proof}
For any $\vs$ and $s'_{\ell}$, let  $\Delta(\vs,s'_{\ell}) = h((\vs,s'_{\ell}))-h(\vs)$.
That is, $\Delta$ is the marginal benefit of adding the intermediate node $s'_\ell$ on top of the nodes $\vs$, assuming that we are greedily attaching all of the edges that we can.

For $i \leq k$, we have
\begin{align*}
h((\vs_i,\vsk)) - h(\vs_i) &= \sum_{j=1}^k \left[ h((\vs_i,\vs^*_j)) - h((\vs_i,\vs_{j-1}^*)) \right] \\
& = \sum_{j=1}^k \Delta((\vs_i,\vs^*_{j-1}),s^*_j) \\
&\leq \sum_{j=1}^k \Delta( \vs_i, s^*_j ),
\end{align*}
where in the last line we have used the fact that the marginal benefit of adding $s^*_j$ later is less than adding it earlier.  (In this sense, $h$ behaves like a submodular function, except that the order of the inputs to $h$ matters; crucially, the function $f$, which is defined on sets rather than sequences, is \em not \em submodular.)
By the definition of \greedyone, we have $\Delta(\vs_i, s_j^*) \leq \Delta(\vs_i, s_{i+1})$ for all $j$, and with the above this implies that
\[ h((\vs_i,\vsk)) - h(\vs_i) \leq \sum_{j=1}^k \Delta(\vs_i, s_{i+1}) = k \cdot \Delta(\vs_i, s_{i+1}). \]

Rearranging this, we have
\begin{equation}\label{eq:incbig}
\Delta(\vs_i, s_{i+1}) \geq \frac{1}{k} \left( h((\vs_i, \vs_k^*)) - h(\vs_i) \right)
\end{equation}
for any $i \leq k$.

Furthermore,
\begin{align}
h((\vs_i,\vsk)) &= h(\vs_{i-1})+\Delta(\vs_{i-1},s_i)+\sum_{j=1}^k\Delta((\vs_i,\vs^*_{j-1}),s^*_j) \notag\\
&\leq h(\vs_{i-1})+\Delta(\vs_{i-1},s_i)+\sum_{j=1}^k\Delta((\vs_{i-1},\vs^*_{j-1}),s^*_j) \label{eq:int}
\end{align}
where in the second line we have used the fact that 
\[ \Delta((\vs_i,\vs^*_{j-1}),s^*_j) \leq \Delta((\vs_{i-1},\vs^*_{j-1}),s^*_j) \]
for any $j$.
Thus, we have
\[ h((\vs_i,\vsk)) \leq h((\vs_{i-1},\vsk)) +\Delta(\vs_{i-1},s_i)\]
using the fact the the right hand side above is equal to the second line of \eqref{eq:int}.
Rearranging, this establishes
\begin{equation}\label{eq:hbig}
 h((\vs_{i-1},\vsk)) \geq  h((\vs_i,\vsk))-\Delta(\vs_{i-1},s_i) 
\end{equation}
Plugging \eqref{eq:hbig} into \eqref{eq:incbig}, we obtain
\begin{align*}
\Delta(\vs_{i-1},s_i) & \geq \frac{1}{k}(h((\vs_i,\vsk))-\Delta(\vs_{i-1},s_i)-h(\vs_{i-1})) 
\end{align*}
and rearranging this implies that
\begin{equation}\label{eq:5}
 \Delta(\vs_{i-1},s_i)  \geq \frac{1}{k+1}(h((\vs_i,\vsk))-h(\vs_{i-1})).
\end{equation}

Now we have
\begin{align*}
 h((\vs_i,\vsk)) - h(\vs_i) &= h((\vs_i,\vsk))-h(\vs_{i-1})-\Delta(\vs_{i-1},s_i) \\
&\leq  h((\vs_i,\vsk)) - h(\vs_{i-1})-\frac{1}{k+1}(h((\vs_i,\vsk))-h(\vs_{i-1})) \\
& = \left(1-\frac{1}{k+1}\right)( h((\vs_i,\vsk)) - h(\vs_{i-1}))
\end{align*}
\end{proof}
where we have used \eqref{eq:5} in the second line.

Finally, we can prove Theorem~\ref{thm:mainapprox}.
\begin{proof}[Proof of Theorem~\ref{thm:mainapprox}]
From Lemma~\ref{lem:submod2}, we have
\[h((\vs_i,\vsk)) - h(\vs_i) \leq \left(1-\frac{1}{k+1}\right) (h((\vs_i,\vsk))-h(\vs_{i-1}))\]
so
\begin{align*}
\left[ h(\vsk)- h(\vs_i)\right] + \left[h((\vs_i,\vsk))-h(\vsk) \right ]  &\leq \left(1-\frac{1}{k+1}\right) \left[h((\vs_i,\vsk))-h(\vs_{i-1})\right] \\
&= \left(1-\frac{1}{k+1}\right) \left[ h(\vsk)-h(\vs_{i-1})\right]+\left(1-\frac{1}{k+1}\right) \left[h((\vs_i,\vsk))-h(\vsk)\right].
\end{align*}
Rearranging, this implies that
\[h(\vsk) -h(\vs_i) \leq  \left(1-\frac{1}{k+1}\right) \left[ h(\vsk)-h(\vs_{i-1})\right] - \frac{h((\vs_i,\vsk))-h(\vsk)}{k+1}\]
Using Lemma~\ref{lem:matchingdiff}, we see that, for all $i$,
\begin{equation}
h(\vsk) -h(\vs_i) \leq \left(1-\frac{1}{k+1}\right)\left[h(\vsk)-h(\vec{s_{i-1}})\right]+\frac{h(\vsk)}{k+1}\cdot \frac{d-1}{d} \label{eq:6}
\end{equation}

Now suppose by induction that
\[ 
h(\vsk) - h(\vs_{i-1}) \leq 
\left( 1 + \frac{1}{d}\left( \left( 1 - \frac{1}{k+1}\right)^{i-1} - 1\right) \right) h(\vsk) \]
The base case for $i=1$ clearly holds. 
Plugging this inductive hypothesis into \eqref{eq:6},
\begin{align*}
h(\vsk) -h(\vs_i) &\leq \left(1-\frac{1}{k+1}\right)\left[h(\vsk)-h(\vec{s_{i-1}})\right]+\frac{h(\vsk)}{k+1}\cdot \frac{d-1}{d} \\
&\leq \left(1-\frac{1}{k+1}\right)
\left( 1 + \frac{1}{d}\left( \left( 1 - \frac{1}{k+1}\right)^{i-1} - 1\right) \right) h(\vsk)+\frac{h(\vsk)}{k+1}\cdot \frac{d-1}{d} \\
&= \left( 1 + \frac{1}{d}\left( \left( 1 - \frac{1}{k+1}\right)^{i} - 1\right) \right) h(\vsk),
\end{align*}
which establishes the inductive hypothesis for $i$.
By induction, we conclude that
\begin{align*}
h(\vsk) - h(\vs_{k}) &\leq 
\left( 1 + \frac{1}{d}\left( \left( 1 - \frac{1}{k+1}\right)^{k} - 1\right) \right) h(\vsk) \\
&\leq \left( 1 + \frac{1}{d} \left( \frac{1}{e} - 1 \right) \right) h(\vsk).
\end{align*}
Rearranging, we have
\[ h(\vs_k) \geq \frac{1}{d}\left( 1 - \frac{1}{e} \right)h(\vsk) ,\]
as desired.
\end{proof}

\section{Experimental Results}\label{sec:experiments}

We first show that multi-layer HAG graphs do not have a significantly higher value for small  $k$ compared to single-layer HAG graphs; this justifies our focus on single-layer HAG graphs in Theorem~\ref{thm:mainapprox}.  We compared \greedyone\ single-layer and multi-layer results for three datasets: a Facebook dataset~\cite{mcauley2012learning}, an Amazon co-purchases dataset~\cite{leskovec2007dynamics} (the subset from March 2nd, 2003), and the Email-EU dataset~\cite{leskovec2007graph}\footnote{All three of these datasets can be found at \UrlFont{snap.stanford.edu/data}}. On average over $k=1,...,100$, the multi-layer results increased  the value compared to the single-layer solution by $3.2\%$, $0.22\%$, and $4.9\%$, respectively (see Table~\ref{tab:singlevsmulti}).

\begin{table}[]
\centering
\begin{tabular}{|l||l|l|l|}
\hline
Dataset                                              & Facebook & Amazon   & Email-EU \\ \hline
Mean value for single-layer HAG                      & 8636.09  & 1800.73  & 3088.73  \\ \hline
Mean value for multi-layer HAG                       & 8945.83  & 1806.29  & 3260.11  \\ \hline
Mean \% improvement for multi-layer HAG              & 3.2\%    & 0.22\%   & 4.9\%    \\ \hline
Std. dev. of \% improvement for multi-layer HAG & 1.02782  & 0.216026 & 1.674153 \\ \hline
\end{tabular}\caption{The improvement of multi-layer over single-layer for \greedyone\ on real-world datasets averaged over $k=1,...,100$.}\label{tab:singlevsmulti}
\end{table}

We next show how well single-layer \greedyone\ and \greedytwo\ perform compared to the optimal single-layer solution (computing the optimum is only tractable for limited graph parameters even in the single-layer case, so we did not implement it for multi-layer HAGs). Figure~\ref{fig:accuracygraphs} shows the quantity $1 - \alpha$, where $\alpha$ is the approximation ratio $\val(\chG_{greedy})/\val(\chG_{opt})$, where $\chG_{greedy}$ is the solution returned by for \greedyone\ and \greedytwo, and $\chG_{opt}$ is the optimal solution, for Erd\H{o}s-R\'enyi graphs $G(n,p)$ with $n=15$ and various values of $p$. Higher values of $p$ result in approximation ratios slightly further from $1$ for both $k=2$ and $k=3$, although in all experiments the approximation ratios are quite close to $1$ for both algorithms. 

\begin{figure}
\centering
 \subfloat[\label{fig:comparek2n}]{\includegraphics[width=.3\textwidth]{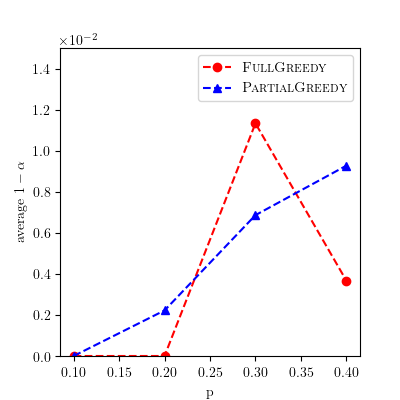}}
 \subfloat[\label{fig:comparek2p}]{\includegraphics[width=.3\textwidth]{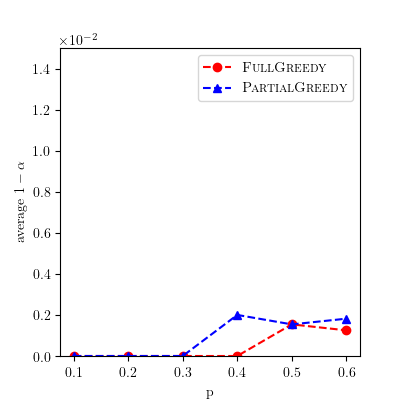}}

\caption{We compare \greedyone\ and \greedytwo\ to the optimal HAG computation graph on a set of $50$ Erd\H{o}s-R\'enyi graphs $G(n,p)$ with $n=15$. The $y$-axis plots average values of $1-\alpha$, where $\alpha$ is the approximation ratio. The $x$-axis plots the parameter $p$.  Shown are (a) $k=2$ and (b) $k=3$. }
\label{fig:accuracygraphs}
\end{figure}

\subsection{Faster Heuristics}\label{sec:hubexp}
While \greedyone\ and \greedytwo\ are much faster in practice than computing the optimal solution, they are still computationally intensive for large values of $k$ and large datasets. In this section we describe two alternative heuristics, \degalgo\ and \hubalgo, which only achieve a fraction of  the value of \greedyone, but compute the HAG computation graph significantly faster.

\degalgo\ starts by ranking all of the vertices of the input graph $G=(V,E)$ by degree:  $\{v_i\}_{i=1}^n$ with  $\Gamma_{out}(v_i)\geq \Gamma_{out}(v_{i+1})$ for $i=1,...,n$. It then takes the top $k$ adjacent pairs of the sequence (i.e., $(v_1,v_2),(v_2,v_3),\ldots,(v_{2k-1},v_{2k})$) as the covers of the $k$ aggregation nodes and constructs a single-layer 2-HAG computation graph. The out-edges of the aggregation nodes are assigned greedily in the same cover order $(v_1,v_2),(v_2,v_3),...$ based on degree. We compare this heuristic to \greedyone\ for value and runtime in Table~\ref{table:heuristicdat}. This method performs decently on the Facebook and Email-EU datasets, and significantly worse on the Amazon purchasing network. We conjecture that this is because the Amazon network has has a significantly lower average degree (about 2.8) than the other two sets (about 22 for Facebook and 25 for Email-EU).  

\hubalgo\ is based on searching  for ``good'' intermediate aggregation nodes around high-degree nodes of $G$. This algorithm is motivated by the frequency with which triangles appear in real-datasets. \hubalgo\ also starts by ranking the vertices from highest to lowest degree as $\{v_i\}_{i=1}^n$. Then for $v_1,...,v_k$ the heuristic does the following: for each $u \in \Gamma_{in}(v_i)$, compute the value of adding aggregation node with cover $\{v_i,u\}$. Then a new node $m$ is added with cover $\{v_i,u\}$ using the $u$ that allows for maximal out-edges from $m$. This process is repeated for $v_1,..,v_k$ in order, so it is greedy in the sense that out neighbors of previous aggregation nodes remain the same during subsequent iterations. We compare \hubalgo\ to \greedyone\ for value and runtime, shown in Table~\ref{table:heuristicdat}.  

\begin{table}[]\centering
\begin{tabular}{l|llll}
         & \multicolumn{2}{l|}{\degalgo \text{ vs. }\greedyone}              & \multicolumn{2}{l}{\hubalgo \text{ vs. }\greedyone} \\
Dataset  & Value Ratio & \multicolumn{1}{l|}{ Runtime  Ratio } &  Value Ratio    & Runtime Ratio     \\ \hline
Amazon   & 0.0699           & \multicolumn{1}{l|}{0.123}           &         0.629               & 0.124
                \\
Email-EU & 0.558             & \multicolumn{1}{l|}{0.0548}           & 0.410                & 0.107
               \\
Facebook & 0.376             & \multicolumn{1}{l|}{0.0408}                                  & 0.313   & 0.0894              
\end{tabular}\caption{For each dataset, \greedyone, \degalgo\, and \hubalgo\ were run 10 times with $k=100$. Value Ratio is computed as the value of the \degalgo\ result divided by the value of the \greedyone\ result for the first column and the value of \hubalgo\ result divided by value of \greedyone\ for the third column. Runtime Ratio is computed in the same way to compare the two heuristics  to \greedyone. 
}\label{table:heuristicdat}
\end{table}

In this paper we have analyzed the optimization problem that arises from \em Hierarchical Aggregation \em (HAG), as introduced by~\cite{HAG} for speeding up learning on GNNs. We showed that \greedyone, the algorithm proposed by~\cite{HAG}, cannot do better than a 1/2 approximation. We also described a second greedy algorithm, \greedytwo, which can actually be implemented efficiently for some parameters, and can obtain results  strictly better than \greedyone. We also showed that \greedyone\ achieves a $\frac{1}{d}(1-1/e)$ approximation ratio for a related objective function where $d$ is the in-degree of the intermediate aggregation nodes. 

Next, we showed empirically that single-layer HAGs achieve nearly the same value as multi-layer HAGs and \greedyone\ and \greedytwo\ both get fairly close to the optimal value on small synthetic graphs. Finally, we defined two additional greedy heuristics, \degalgo\ and \hubalgo, and showed that they can achieve about a third to a half of the value of \greedyone\ in a tenth or less of the runtime. 

Our work suggests many interesting future directions, including pinning down the approximation ratio for both \greedyone\ and \greedytwo, and proving approximation guarantees for the heuristics \degalgo\ and \hubalgo\ in terms of the characteristics of the graph.

\section*{Acknowledgements} 
We thank Zhihao Jia, Rex Ying, and Jure Leskovec for helpful conversations.

\bibliographystyle{plain}
\bibliography{citations}

\end{document}